\documentclass{article}

\usepackage[dblblindworkshop, final]{neurips_2026}
\workshoptitle{Dynamic Alignment on Human-AI Coupled System (BiAlign Workshop)}

\usepackage[utf8]{inputenc} 
\usepackage[T1]{fontenc}    
\usepackage{hyperref}       
\usepackage{url}            
\usepackage{booktabs}       
\usepackage{amsfonts}       
\usepackage{nicefrac}       
\usepackage{microtype}      
\usepackage{xcolor}         

\usepackage{amsmath, amssymb, amsthm, graphicx, changepage, enumitem, caption, makecell, tikz}
\usetikzlibrary{arrows.meta, positioning}

\newtheorem{theorem}{Theorem}[section]
\newtheorem{proposition}[theorem]{Proposition}

\newtheorem{definition}[theorem]{Definition}

\newtheorem{remark}[theorem]{Remark}

\title{From Observation to Intervention: Identifiability of Latent User Bias in Human–AI Interaction}

\author{%
  Lucas Biechy \\
  Inria \\
  INSA-CVL \\
  MARS \\
  \And
  Yuxiao Li \\
  Independent \\
  MARS \\
  \And
  Zhonghao He \\
  Oxford Human-Centered AI Lab \\
  Johns Hopkins University \\
  MARS \\
  \And
  Tianyi Qiu \\
  Stanford University \\
  MARS \\
}

\begin{document}

\maketitle

\begin{abstract}
Alignment methods typically treat user behavior as a reliable proxy for underlying preferences. However, this assumption breaks down when users hold belief-dependent cognitive biases, risking irreversible ``confirmation traps'' and harmful feedback loops. In this paper, we formalize when an AI system can distinguish true user preferences from behavior shaped by cognitive bias. Modeling this interaction as a multi-armed bandit, we first prove a non-identifiability theorem: under passive observation, a static binary biased user produces action trajectory distributions strictly identical to an unbiased user in a transformed environment, forcing any action-only estimator to perform no better than random guessing. We then show that active preference elicitation achieves identifiability when the bias is belief-dependent, as confirmation bias is: an observer applying a finite sequence of adversarial interventions makes the post-intervention action distributions arbitrarily separable. This identifiability result applies specifically to belief-dependent confirmation biases and not static ones. Our theoretical results suggest that robust alignment requires actively intervening to tell informed preferences from biases.
\end{abstract}

\section{Introduction}

AI systems that adapt to a user increasingly treat that user's behavior as ground truth for underlying preferences. Reinforcement learning from human feedback (RLHF) trains language models directly on human choice dynamics \citep{christiano2017deep}. Scalable oversight research aims to extend this role to settings where human supervisors cannot fully verify the target system \citep{bowman2022measuring, irving2018ai}. However, this practice rests on an assumption rarely tested: that a human agent's actions faithfully reveal its internal preferences. Building upon prior work that highlights the general fragility of this proxy \citep{casper2023open} and its vulnerability to feedback loops \citep{chaney2018how, jiang2019degenerate, krueger2020hidden}, we study the case where the user suffers from a systematic cognitive bias, updating honestly on distorted experiences. Confirmation bias is the primary instance of such distortion \citep{nickerson1998confirmation, rabin1999first, epley2016mechanics}, with well-documented behavioral and neural manifestations in human reinforcement learning and self-assessment \citep{palminteri2017confirmation, lefebvre2017behavioural}. An observer trying to infer hidden preferences from action channels alone is reading a signal corrupted by the very mental state it seeks to recover.

Prior work infers preferences mainly from learning or exploring agents
\citep{jacq2019learning, balakrishna2019onpolicy, brown2019extrapolating, ramponi2020inverse, guo2021learning, abdelsadek2026learning}, whereas demonstrator suboptimality is merely transient and fades with more observation (\ref{sec:related_work}).
A cognitive bias acts as a permanent, belief-dependent filter through which the user's posterior converges to a distorted belief. This raises a fundamental question: \emph{can an observer restricted purely to action trajectories ever distinguish a cognitively biased user from an unbiased one?}

\begin{figure}[t]
    \centering
    \includegraphics[width=1\linewidth]{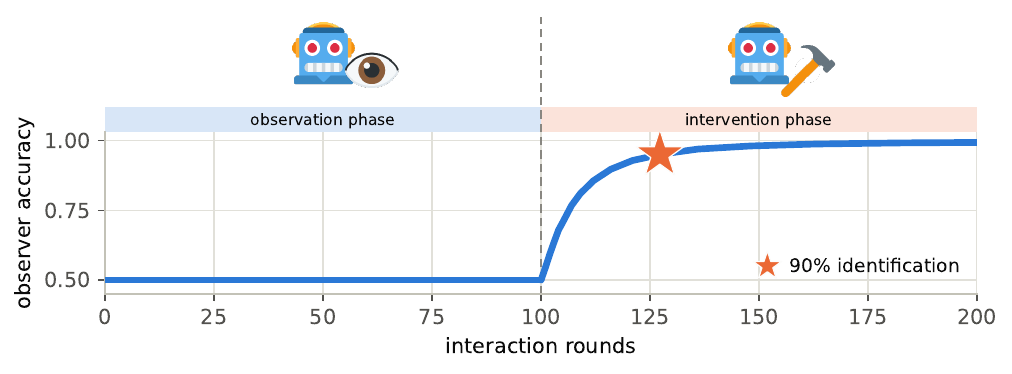}
    \caption{\textbf{Passive observational equivalence vs. Active identification.} --- Passive observation cannot distinguish a biased user in a world $\boldsymbol{\mu}$ from an unbiased user in world $\boldsymbol{\nu}$ ($P_{M_1}^{a,T} = P_{M_0}^{a,T}$). For a belief-dependent (confirmation) bias, a block of zero-reward interventions restores identifiability (Theorem~\ref{thm:identifiability}).}
    \label{fig:hero-plot}
\end{figure}
We show that passive behavioral observation fails in two severe ways. First, bias is not a transient inconvenience. Under a belief-dependent cognitive kernel that attenuates negative feedback and anchors prior expectations (Definitions~\ref{def:conf_bias}), Thompson Sampling \citep{chapelle2011empirical, agrawal2012analysis, kaufmann2012thompson, russo2018tutorial} plays a strictly suboptimal arm forever with positive probability, collapsing logarithmic regret into linear regret (Theorem~\ref{thm:lock-in}). Though non-Bayesian updates in RL agents can yield good empirical effects when having state access \citep{shen2024cmdqn}, our result provides a worst-case guarantee where confirmatory updating creates an unrecoverable trap assuming no state access. Second, this is an instance of observational equivalence in the causal-inference sense \citep{verma1990equivalence, manski2003partial, bareinboim2015bandits}. By constructing a static bias (Definition~\ref{def:static}), we show that a biased user and an unbiased user in a rescaled environment generate \emph{identical} trajectory distributions at every horizon (Theorem~\ref{thm:non-identifiability}). Thus, action-only estimators cannot beat the random choice. Unlike reward-recovering estimators in standard learning settings \citep{guo2021learning}, this observational confound is structural and permanent.

Fortunately, active intervention achieves identifiability for belief-dependent biases (Figure \ref{fig:hero-plot}). Mechanically, this mirrors adversarial reward corruption in bandits \citep{lykouris2018stochastic, gupta2019better, bogunovic2021stochastic}, but with the intent reversed: corruption is no longer an adversary's tool against the learner, but an instrument to break an observational confound \citep{eberhardt2007interventions, hyttinen2013experiment, lattimore2016causal, rainforth2024modern}. Paying a small cost to force objective failures separates user types within a finite, computable budget $m_\delta$ (Theorem~\ref{thm:identifiability}), even under softer protocols that never override spontaneous user choices (Proposition~\ref{appendix:sequantial-id}). Finally, we evaluate these theoretical limits empirically (\S\ref{sec:exp-trap}--\S\ref{sec:exp-active}) and probe how real language models behave in such interaction loops (\S\ref{sec:exp-llm}).

Our key contributions are:

\begin{itemize}

    \item \textbf{Decoupled Formalism For Cognitive Bias:} A formal setup decoupling Bayesian updates from cognitive distortion, isolating cognitive bias from reasoning failure (\S\ref{sec:setup}).

    \item \textbf{Confirmation Trap Analysis:} Proof that belief-dependent cognitive filters turn logarithmic Thompson Sampling into linear regret with positive probability (Theorem~\ref{thm:lock-in}).

    \item \textbf{Non-Identifiability:} Theoretical demonstration that biased and unbiased users across distinct environments are permanently indistinguishable from actions alone under static channel bias (Definition~\ref{def:static}, Theorem~\ref{thm:non-identifiability}).

    \item \textbf{Active Remediation Protocols:} Constructive intervention schemes using reward perturbations to restore identifiability under coercive (Theorem~\ref{thm:identifiability}) and non-coercive setups (Proposition~\ref{appendix:sequantial-id}), a characterization of when they fail for static channels (Remark~\ref{rem:static-blind}), and an explicit map of which result holds under which bias model (Table~\ref{tab:scope}), validated by simulations (\S\ref{sec:exp-trap}--\S\ref{sec:exp-active}). LLM experiments further test whether \texttt{Qwen3-4B} also exhibit the theoretical results (\S\ref{sec:exp-llm}).

\end{itemize}

\section{Related Work}
\label{sec:related_work}

\paragraph{Behavioral proxies for preference and learning from imperfect demonstrators.}
Alignment pipelines train reward models directly from human preference comparisons \citep{christiano2017deep}. Scalable oversight techniques like sandwiching \citep{bowman2022measuring} and debate \citep{irving2018ai} rely on this same assumption that user choices reflect true preferences. When raters exhibit systemic biases, literature often treats these issues as data quality flaws that can be fixed through better auditing and data collection \citep{casper2023open}. Similarly, work on learning rewards from imperfect demonstrators assumes that human suboptimality stems from temporary exploration or active learning that clears up with additional data \citep{jacq2019learning, balakrishna2019onpolicy, ramponi2020inverse, brown2019extrapolating, guo2021learning, abdelsadek2026learning}.
Closest to our question, \citep{chan2021humanirrationalitybadgood} illustrates  that a fixed, known irrationality model can help or hurdle reward inference. Instead, we ask the prior and fundamental question that whether the bias model is identifiable from behavior at all.

\paragraph{Confirmation bias and sycophancy.}
Confirmation bias is well documented in human reasoning and sequential decision-making \citep{nickerson1998confirmation, rabin1999first, garrett2020biased}. In human reinforcement learning, this bias appears as an asymmetric updating rate that favors positive feedback \citep{palminteri2017confirmation, lefebvre2017behavioural, palminteri2022computational}. Standard alignment protocols aggravate this human bias because they treat the user as an omniscient oracle rather than a learner who may hold incorrect beliefs. To maximize immediate human approval, language models adopt sycophantic behaviors by mirroring the user's initial assumptions \citep{perez2022discovering, sharma2024towards, denison2024sycophancy}. This dynamic creates a self-reinforcing loop: the sycophantic model validates the user's biased expectations, preventing disconfirming feedback and locking the user inside a confirmation trap.

\paragraph{Bandit algorithms and distorted perceptions.}
Unbiased Thompson Sampling matches classical logarithmic regret lower bounds \citep{lai1985asymptotically, kaufmann2012thompson, agrawal2012analysis, chapelle2011empirical, russo2018tutorial}. Deterministic baselines such as Upper Confidence Bound (UCB) achieve the same rate without maintaining a belief state to distort \citep{auer2002finite}. Existing analyses of Thompson Sampling under misspecified priors show only bounded regret degradation \citep{liu2016prior, simchowitz2021bayesian}. Similarly, behavioral bandits with statically distorted value perception still achieve sublinear regret \citep{gopalan2017weighted}. In contrast, our work proves an unbounded trap with positive probability caused by belief-dependent, self-reinforcing distortions. This distinction is particularly relevant because human decision-makers are often modeled as Bayesian bandit learners \citep{steyvers2009bayesian}.

\paragraph{Non-identifiability, active intervention, and behavioral feedback loops.}We cast passive non-identifiability causally: a biased user in one environment and an unbiased user in a rescaled one are observationally equivalent \citep{verma1990equivalence}. This partial identification \citep{manski2003partial} stems from a confound inside the agent's perception rather than in the observer's causal estimates \citep{bareinboim2015bandits}. Similar constraints exist in reward learning, where demonstrations and comparisons recover rewards only up to equivalence classes \citep{skalse2023invariance, cao2021identifiability}.Active interventions resolve such confounds \citep{eberhardt2007interventions}, as studied in experiment selection \citep{hyttinen2013experiment}, causal bandits \citep{lattimore2016causal}, and Bayesian experimental design \citep{rainforth2024modern}. Our closed-form horizon $m_\delta$ (Proposition~\ref{prop:computable_horizon}) gives a fixed-confidence identification bound \citep{garivier2016optimal} for this hypothesis. Mechanically, our approach repurposes reward-poisoning techniques \citep{lykouris2018stochastic, gupta2019better, bogunovic2021stochastic, jun2018adversarial} as diagnostic tools for active preference elicitation \citep{sadigh2017active, ji2024rlhf, hadfieldmenell2016cooperative}. In contrast, relying solely on passive behavior leads to degenerate recommendation loops \citep{chaney2018how, jiang2019degenerate, krueger2020hidden, milli2025engagement}.
\section{Problem Setup}\label{sec:setup}

\paragraph{MAB Environment \& Perceptual Distortion Framework.}
We consider a Multi-Armed Bandit (MAB) problem with $K$ arms, indexed by $k \in \{1, \dots, K\}$. Each arm is associated with an unknown mean return $\mu_k \in \Theta$, defining the environment state $\boldsymbol{\mu} = (\mu_1, \dots, \mu_K) \in \Theta^K$. At each time step $t \in \{1, \dots, T\}$, the agent selects action $a_t \in \{1, \dots, K\}$ and receives an objective ground-truth reward $r_t \sim P(r \mid \mu_{a_t})$.
We parameterize the decision-maker via a latent user type $z \in \{0, 1\}$. While a reference objective agent ($z=0$) directly processes ground-truth signals $r_t$, a perceptually biased user $\mathcal{A}_{\text{bias}}$ ($z=1$) observes a subjectively corrupted reward $\tilde{r}_t$. This subjective signal is generated by a perceptual distortion kernel that may depend on prior expectations:
\begin{equation*}
    \tilde{r}_t \sim P_{\text{bias}}\left(\tilde{r} \;\middle|\; r_t, \tilde{\mu}_{a_t}^{(t-1)}, z\right)
\end{equation*}
where $\tilde{\mu}_{a_t}^{(t-1)} := \mathbb{E}_{\tilde{b}_{t-1}}[\mu_{a_t}]$ is the prior expected return for arm $a_t$. For $z=0$, perception is identity ($\tilde{r}_t = r_t$). The agent maintains history $\tilde{h}_t := \big((a_i, \tilde{r}_i)\big)_{i=1}^{t-1}$ and updates its posterior via Bayes' rule: $\tilde{b}_t(\boldsymbol{\mu}) \propto \tilde{b}_0(\boldsymbol{\mu}) \prod_{i=1}^{t-1} P(\tilde{r}_i \mid a_i, \mu_{a_i})$, treating subjective signals as objective.

\paragraph{Behaviorally Grounded Bias Instantiations.}
While $P_{\text{bias}}$ accommodates arbitrary distortion distributions, establishing theoretical bounds requires specifying how beliefs distort incoming signals. We therefore ground our analysis in two canonical mechanisms. They answer different questions: the belief-dependent kernel (Def.~\ref{def:conf_bias}) captures the dynamics of confirmation bias, while the memoryless channel (Def.~\ref{def:static}) is the simplest bias that can be exactly confounded with the environment. Table~\ref{tab:scope} states which result holds under which model.

First, to capture classical confirmation bias where individuals selectively discount evidence contradicting prior expectations \citep{nickerson1998confirmation, rabin1999first, palminteri2017confirmation}, we model belief-dependent distortion via prediction-error attenuation and its non-linear Bernoulli manifestation near certainty \citep{lefebvre2017behavioural}.

\begin{definition}[Belief-Dependent Confirmation Bias]\label{def:conf_bias}
An agent $\mathcal{A}_{\text{bias}}$ exhibits \emph{confirmation bias} if subjective perception dampens disconfirming evidence:
\begin{enumerate}[leftmargin=*, topsep=0pt, partopsep=0pt, parsep=0pt, itemsep=0pt]
    \item \textbf{$\gamma$-Attenuation:} Whenever prior expectation exceeds objective outcome ($\tilde{\mu}_{a_t}^{(t-1)} > r_{t}$), expected perception satisfies:
    \begin{equation}\label{eq:attenuation}
        \mathbb{E}_{\tilde{r} \sim P_{\text{bias}}}\left[ \tilde{r} \;\middle|\; r_t, \tilde{\mu}_{a_t}^{(t-1)} \right] \ge r_t + \gamma \left( \tilde{\mu}_{a_t}^{(t-1)} - r_t \right) \quad \text{for } \gamma \in (0, 1]
    \end{equation}
    \item \textbf{Persistent Confirmation near Certainty:} In Bernoulli settings ($r_t \in \{0, 1\}$), positive outcomes are kept and negative ones suppressed superlinearly with prior confidence, for $C > 0$, and $p > 1$:
    \begin{equation}\label{eq:confirmation}
        P_{\mathrm{bias}}(\tilde r = 1 \mid r = 1, \cdot) = 1,
        \qquad
        P_{\mathrm{bias}}(\tilde r = 0 \mid r = 0, \tilde{\mu}_{a_t}^{(t-1)}) \le C \left(1 - \tilde{\mu}_{a_t}^{(t-1)}\right)^p
    \end{equation}
\end{enumerate}
\end{definition}

Second, to complement state-dependent cognitive distortions with a memoryless baseline, we also consider static channel bias.

\begin{definition}[Static Binary Bias]\label{def:static}
An agent exhibits \emph{static channel bias} with parameters $(u, v) \in [0, 1]^2 \setminus \{(1, 0)\}$ if outcome perception is memoryless and independent of prior beliefs:
\begin{equation}
    P_{\mathrm{bias}}(\tilde r = 1 \mid r = 1) = u, \qquad P_{\mathrm{bias}}(\tilde r = 1 \mid r = 0) = v
\end{equation}
\end{definition}

\paragraph{User Policy \& Observer Agent ($\mathcal{A}_{\text{obs}}$).}
We model the user as a Bayesian decision-maker driving exploration via Thompson Sampling \citep{agrawal2012analysis, russo2018tutorial}. At step $t$, given belief $\tilde{b}_{t-1}(\cdot \mid z)$, the user draws $\boldsymbol{\theta}^{(t)} \sim \tilde{b}_{t-1}(\cdot \mid z)$ and selects $a_t = \arg\max_{k} \theta_k^{(t)}$, with a fresh draw at each round and a fixed tie-breaking rule.
We model the LLM as as an agent $\mathcal{A}_{\text{obs}}$ restricted to observing the user's action history $h_t^a := (a_1, \dots, a_{t-1})$ and its own intervention sequence. This restriction is dual-purposed: ground-truth rewards are unobservable in deployment, while feeding subjective user feedback directly to the model risks inducing sycophancy, where $\mathcal{A}_{\text{obs}}$ blindly reinforces the bias it seeks to detect. The marginal action likelihood evaluated by $\mathcal{A}_{\text{obs}}$ under type $z$ is:
\begin{equation*}
    P(a_t = k \mid h_t^a, z) = \mathbb{E}\Big[\, \pi_{\mathrm{TS}}\big(k \mid \tilde{b}_{t-1}\big) \;\Big|\; h_t^a, z \Big],\qquad \pi_{\mathrm{TS}}(k \mid b) := P_b(\text{TS selects } k)
\end{equation*}

\section{Theoretical Analysis: Passive Limitations}
This section establishes the theoretical boundaries of inferring user preferences and bias under passive observation, demonstrating why standard monitoring fails without active intervention. First, we show that belief-dependent confirmation bias (Def.~\ref{def:conf_bias}) forces Thompson Sampling into an ``confirmation trap'' (Theorem~\ref{thm:lock-in}), degrading logarithmic regret to linear. Second, we prove that static channel biases (Def.~\ref{def:static}) induce an observational equivalence (Theorem~\ref{thm:non-identifiability}) when $\mathcal{A}_{\text{obs}}$ observes action histories $h_t^a$ alone, making passive identification impossible. These two results concern different bias models; Table~\ref{tab:scope} gives the scope of formal results in the paper. Whether passive observation can identify a belief-dependent kernel remains open.

\begin{table}[t]
\centering
\small
\caption{\textbf{Scope of the formal results.} Each bias model answers a different question: the static channel isolates the confound between bias and environment, while the belief-dependent kernel carries the dynamics that interventions exploit.}
\label{tab:scope}
\begin{tabular}{>{\raggedright\arraybackslash}p{0.19\textwidth} >{\raggedright\arraybackslash}p{0.36\textwidth} >{\raggedright\arraybackslash}p{0.36\textwidth}}
\toprule
 & \textbf{Static channel} (Def.~\ref{def:static}) & \textbf{Confirmation bias} (Def.~\ref{def:conf_bias}) \\
\midrule
\textbf{Passive observation} & Non-identifiable at every horizon (Thm.~\ref{thm:non-identifiability}) & Linear regret (Thm.~\ref{thm:lock-in}); identifiability open \\
\textbf{Intervention} & Uninformative when $v=0$ (Remark~\ref{rem:static-blind}) & Identifiable within $m_\delta$ rounds (Thm.~\ref{thm:identifiability}, Prop.~\ref{prop:computable_horizon}, App.~\ref{appendix:sequantial-id}) \\
\bottomrule
\end{tabular}
\end{table}

\paragraph{Confirmation Trap in Biased Agents.}
When cognitive distortion dynamically depends on prior beliefs, negative outcomes are re-interpreted as positive. The following result establishes that this feedback loop causes permanent lock-in on suboptimal actions with positive probability.

\begin{theorem}[Positive-Probability Confirmation Trap]\label{thm:lock-in}
Consider a $K$-armed Bernoulli bandit with true means $\boldsymbol{\mu}\in[0,1]^K$. Suppose a biased agent $\mathcal{A}_{\text{bias}}$ employs Thompson Sampling with independent Beta priors $\tilde b_0(\boldsymbol{\mu})=\prod_{j=1}^K\operatorname{Beta}(\alpha_j,\beta_j)$ with $\alpha_j, \beta_j > 0$. Fix a strictly suboptimal arm $k$ ($\Delta_k := \mu^\star - \mu_k > 0$). Under Negative Prediction-Error Attenuation and Persistent Confirmation (\emph{Definition~\ref{def:conf_bias}}), if the competing arms satisfy:
\begin{equation}
    \beta_{\min}^{(-k)} := \min_{j\neq k}\beta_j > 1
\end{equation}
then there exists a sample paths event $\mathcal{E}_k$ with strictly positive probability ($\mathbb P(\mathcal E_k)>0$) defined by:
\begin{equation}
    \mathcal E_k := \bigcap_{t=1}^{\infty}\{a_t=k, \; \tilde r_t=1\}
\end{equation}
On $\mathcal E_k$, the subjective expectation and selection probability converge to certainty: $\tilde\mu_k^{(t)} \to 1$ and $P(a_{t+1}=k \mid \tilde h_{t+1}) \to 1$. Consequently, the expected pseudo-regret grows linearly: $\mathbb E[R_T] \ge \rho_k \Delta_k T$ for some constant $\rho_k > 0$ independent of $T$.
\end{theorem}

Proof in Appendix~\ref{proof:lock-in}.

Without distortion, Thompson Sampling with the same priors achieves classical logarithmic regret \citep{agrawal2012analysis, kaufmann2012thompson}, but Theorem~\ref{thm:lock-in} shows this guarantee fails under belief-dependent cognitive bias.

\paragraph{Action-Only Non-Identifiability.}
Having established that cognitive bias incurs severe regret, a fundamental question arises: \emph{Can an external observer detect whether a user is biased purely by monitoring their interactions?} Because $\mathcal{A}_{\text{obs}}$ observes neither ground-truth rewards $r_t$ nor perceived rewards $\tilde{r}_t$, cognitive distortion can be perfectly confounded with environment statistics \citep{verma1990equivalence, bareinboim2015bandits}.

\begin{definition}[Action Identifiability]\label{def:identifiability}
Let $\mathcal M_0$ and $\mathcal M_1$ denote the model classes of unbiased and biased user configurations, respectively. User type $z \in \{0, 1\}$ is \emph{action-identifiable} if their induced action distribution spaces are disjoint:
\begin{equation}
    \left\{P_M^{a,T} : M \in \mathcal M_0\right\} \cap \left\{P_M^{a,T} : M \in \mathcal M_1\right\} = \varnothing \quad \text{for some } T \ge 1
\end{equation}
where $P_M^{a,T} := \mathcal L_M(a_1,\ldots,a_T)$ denotes the probability law over action trajectories induced by model $M$. Otherwise, the user type is \emph{action-non-identifiable}.
\end{definition}

For every biased user, \textit{there exists an honest world that explains their behavior exactly}. Watching actions alone can never tell the two apart (Figure~\ref{fig:hero-plot}).

\begin{theorem}[Fundamental Passive Non-Identifiability]\label{thm:non-identifiability}
Consider a $K$-armed Bernoulli bandit with unknown means $\boldsymbol{\mu} \in [0,1]^K$. Let $M_1 \in \mathcal{M}_1$ be a biased agent subject to Static Binary Bias (\emph{Definition~\ref{def:static}}) with effective arm means $\nu_k := u\mu_k + v(1-\mu_k)$. Let $M_0 \in \mathcal{M}_0$ be an unbiased agent operating in an environment with true means $\boldsymbol{\nu} = (\nu_1, \dots, \nu_K)$.

If $M_0$ and $M_1$ share prior $b_0 = \tilde{b}_0$ and the same Thompson Sampling rule (fresh posterior draw each round, independent of rewards and perception noise; common tie-breaking rule), they induce identical action trajectory distributions for all $T \ge 1$:
\begin{equation}
    P_{M_1}^{a,T} = P_{M_0}^{a,T}
\end{equation}
Consequently, no action-only estimator $\widehat z_T(a_1, \dots, a_T)$ can beat random guessing:
\[
\max \{ P_{M_0}(\widehat z_T = 1), P_{M_1}(\widehat z_T = 0) \} \ge \frac{1}{2}
\] 
Furthermore, for any prior $\pi_1 = P(M_1) \in (0, 1)$, passive observation yields zero information gain, keeping posterior odds strictly static: $P(M_1 \mid h_T^a) = \pi_1$ for all $T \ge 1$.
\end{theorem}

Proof in Appendix~\ref{proof:non-identifiability}.

\section{Active Interventions and Identification}\label{sec:active-id}

Having established that passive observation cannot distinguish cognitive bias from different environmental conditions (Theorem~\ref{thm:non-identifiability}), we now show how an active Observer can identify a belief-dependent bias (Def.~\ref{def:conf_bias}), and why the same test can fail for a static channel. By forcing deterministic objective failures, the Observer induces divergent behaviors. An unbiased agent learns the failure rate and abandons the arm, whereas a biased agent filters disconfirming signals and continues to play it.

\paragraph{Adversarial Zero-Reward Interventions.}
Before a selected action yields a natural outcome, the Observer can choose to actively alter the reward pipeline.

\begin{definition}[Adversarial Zero-Reward Interventions \& Block Protocol]\label{def:interventions}
At round $t$, upon user choice $a_t$, the observer selects an intervention $I_t \in \{0, 1\}$ yielding delivered objective reward
\[
r_t = (1 - I_t) r_t^\dagger
\]
where $r_t^\dagger \sim \operatorname{Bernoulli}(\mu_{a_t})$. Forcing an objective failure ($I_t = 1$) incurs cost $c > 0$, with cumulative cost $C_T = c \sum_{t=1}^T I_t$. The observer tracks action history $h_t^a$ and its own past interventions $\{I_i\}_{i=1}^{t-1}$, but sees neither delivered rewards $r_t$ nor subjective $\tilde{r}_t$.

To systematically force belief revision, a \emph{controlled zero-reward block} of length $m$ on arm $k$ starting at time $\tau$ enforces $a_\tau = \dots = a_{\tau+m-1} = k$ with active interventions $I_\tau = \dots = I_{\tau+m-1} = 1$.
\end{definition}

\paragraph{Active Bias Identification.}
By delivering $m$ consecutive zero-rewards, the Observer creates a regime where the unbiased model $M_0$ updates its posterior toward zero, while the biased model $M_1$ preserves its prior expectation under Persistent Confirmation (\emph{Eq. ~\ref{eq:confirmation}}).

\begin{theorem}[Active Identification via Zero-Reward Block]\label{thm:identifiability}
Fix arm $k$ and suppose candidate models $M_0, M_1$ share independent Beta posteriors at step $\tau$, with $\mu_k \mid \mathcal{F}_\tau \sim \operatorname{Beta}(A, B)$. Define the tail probability bound $\varepsilon_\infty(A,B) := C B^p \sum_{\ell=0}^{\infty} (A+B+\ell)^{-p}$. Under Persistent Confirmation (\emph{Eq.~\ref{eq:confirmation}}), if the initial state satisfies $\varepsilon_\infty(A,B) \le \frac{\delta}{6}$ for target error $\delta \in (0,1)$, then there exists a finite block length $m_\delta$ such that for all $m \ge m_\delta$:
\begin{equation}
    \left\| Q_1^{(m)} - Q_0^{(m)} \right\|_{\mathrm{TV}} > 1 - \delta
\end{equation}
where $Q_z^{(m)}$ denotes the post-intervention Thompson Sampling action $a_{\tau+m}$ distribution under type $z \in \{0, 1\}$, marginalized over latent perceived rewards $\tilde{r}_t$ under $M_1$.

Consequently, the single-action decision rule $\widehat{z} = \mathbf{1}\{a_{\tau+m} = k\}$ identifies user bias with error at most $\delta$, at total intervention cost $C = c \, m_\delta$.
\end{theorem}

Proof in Appendix~\ref{proof:identifiability}.

The active identification test has a minimal length $m_\delta$, explicitly computable in closed form directly from the observer's prior.

\begin{proposition}[Computable Intervention Horizon]\label{prop:computable_horizon}
Fix target error rate $\delta \in (0,1)$ and set threshold $\xi := \delta / 6$. Let $G_j$ denote the CDF of the Thompson sample from competitor arm $j \neq k$. 
Choose thresholds $x, y \in (0,1)$ such that $G_{j_0}(x) \le \xi$ for some $j_0 \neq k$ and $\sum_{j \neq k} [1 - G_j(1-y)] \le \xi$. Define the tail bound functions:
\begin{align}
    L_0(m; x) &:= D_A(x) \frac{\Gamma(A+B+m)}{\Gamma(A)\Gamma(B+m+1)} (1-x)^{B+m} \\
    L_1(m; y) &:= D_B(y) \frac{\Gamma(A+B+m)}{\Gamma(B)\Gamma(A+m+1)} (1-y)^{A+m}
\end{align}
where $D_\alpha(s) := \max\{1, s^{\alpha-1}\}$. Then, a sufficient block length for Theorem~\ref{thm:identifiability} is given in closed form by:
\begin{equation}
    m_\delta := \min \left\{ m \ge 1 : L_0(m; x) \le \xi \quad \text{and} \quad L_1(m; y) \le \xi \right\}
\end{equation}
\end{proposition}

While Theorem~\ref{thm:identifiability} relies on forced action assignment, we extend this framework in Appendix~\ref{appendix:sequantial-id} to non-coercive settings where the Observer cannot force actions and intervenes strictly upon spontaneous arm selection. 
\textit{Crucially, the persistent self-confirmation that causes lock-in under passive observation is what betrays the biased agent under non-coercive interventions: while an unbiased user quickly abandons the zeroing arm, a biased user continues re-selecting it.}

\begin{remark}[Static channels can be blind to zero-reward interventions]\label{rem:static-blind}
The proof of Theorem~\ref{thm:identifiability} relies on Persistent Confirmation, which a static channel (Def.~\ref{def:static}) satisfies only if $u=v=1$. When $v=0$, any zero-reward intervention policy, coercive or not, that uses only the observer's information and independent randomization leaves the joint law of actions and interventions identical under the biased user in $\boldsymbol{\mu}$ and the unbiased user in $\boldsymbol{\nu}$ of Theorem~\ref{thm:non-identifiability}, so the confound survives (Appendix~\ref{proof:static-blind}).
\end{remark}
\section{Experiments}
\label{sec:experiments}

We evaluate our theoretical framework in Bernoulli bandits under Beta priors using Thompson Sampling, testing the boundary-tight dynamic kernel (\S\ref{sec:exp-setup}) and the static binary bias (Def. \ref{def:static}). Table~\ref{tab:overview} in Appendix~\ref{sec:exp-appendix} summarizes our findings.

\subsection{Experimental Setup}
\label{sec:exp-setup}

Every experimental instance evaluates a $K$-armed Bernoulli bandit with independent Beta priors under Thompson Sampling. Biased agents ($z=1$) operate under the canonical boundary-tight distortion kernel (\S\ref{sec:setup}), defined by $P_{\text{bias}}(\tilde r{=}1 \mid r{=}1, \tilde\mu) \equiv 1$ and:
\begin{equation}\label{eq:bias_xp}
P_{\text{bias}}(\tilde r{=}1 \mid r{=}0, \tilde\mu) = g(\tilde\mu) := \min\Big(1, \max\big(\gamma\tilde\mu, 1-C(1-\tilde\mu)^p\big)\Big)
\end{equation}
where $(\gamma, C, p)$ are parameters. This equation represents the minimal distortion satisfying~\ref{def:conf_bias}. 

We benchmark the coercive intervention block (Definition~\ref{def:interventions}) against the non-coercive identification policy (Proposition~\ref{appendix:sequantial-id}), implemented as distinct environment wrappers. All simulations run on a vectorized seed-batched framework across $10^4$--$2\times10^4$ independent seeds. Theoretical quantities---including the lock-in recursion probability $\rho_k$ (Theorem~\ref{thm:lock-in}), the intervention horizon $m_\delta$ (Theorem~\ref{thm:identifiability}), and posterior odds (Theorem~\ref{thm:non-identifiability})---are evaluated via exact numerical integration.\footnote{We report 95\% Wilson score confidence intervals for proportions, bootstrap confidence intervals for Total Variation (TV) distances, and standard errors (SEM) for continuous regret metrics.}

\subsection{The Confirmation Trap Is Real and Suboptimal}
\label{sec:exp-trap}

\textbf{Takeaway 1 (Bias is expensive)}: By $T{=}1{,}000$, the biased agent incurs $17\times$ the regret of its unbiased twin, while our theoretical recursion predicts empirical lock-in frequencies to within $\pm 0.008$

\begin{figure}[htbp]
\centering
\begin{minipage}{0.34\linewidth}
\centering
\includegraphics[width=\linewidth]{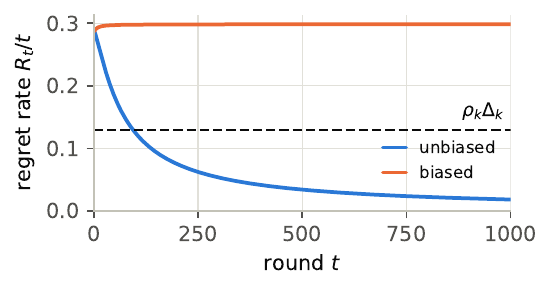}
\end{minipage}\hfill
\begin{minipage}{0.65\linewidth}
\centering
\includegraphics[width=\linewidth]{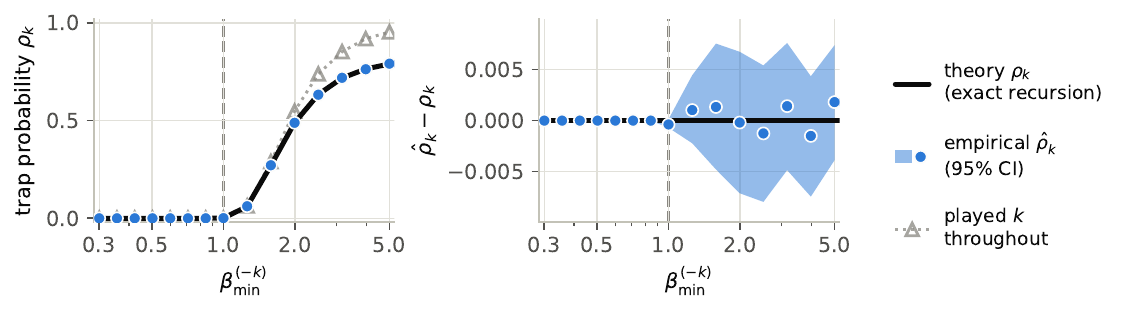}
\end{minipage}
\caption{\textbf{Empirical validation of confirmation trap dynamics (Theorem~\ref{thm:lock-in}).} \\
\textbf{Left:} Cumulative regret rate $R_t/t$ for $z=0$ vs.\ $z=1$ ($4$-arm bandit, $\boldsymbol\mu=(0.5,0.8,0.6,0.3)$, target prior $\mathrm{Beta}(2,1)$, competing priors $\mathrm{Beta}(1,3)$, kernel $(\gamma,C,p)=(0.5,1,2)$, $4{,}000$ seeds, $95\%$ SEM). Dashed line indicates lower bound $\rho_k\Delta_k $. \\
\textbf{Right:} Theoretical $\rho_k$ vs.\ empirical $\hat{\rho}_k$ over $\beta_{\min}^{(-k)}$ ($2$-arm bandit, $\boldsymbol\mu=(0.3,0.7)$, $20{,}000$ seeds/point, Wilson $95\%$ CI). Vertical dashed line marks $\beta_{\min}^{(-k)}=1$; triangles represent $\{a_t{=}k\ \forall t\}$.}
\label{fig:trap}
\end{figure}

Theorem~\ref{thm:lock-in} establishes that under $\beta_{\min}^{(-k)}>1$, a biased Thompson Sampling agent locks into a suboptimal arm with probability $\rho_k > 0$, incurring linear regret bounded below by $\mathbb{E}[R_t] \ge \rho_k \Delta_k t$. As shown in Figure~\ref{fig:trap} (left), the unbiased agent achieves sublinear logarithmic regret ($R_t/t \to 0$), whereas the biased agent plateaus at an empirical regret rate $R_t/t \approx 0.3$. This rate strictly exceeds the theoretical lower bound ($\rho_k \Delta_k \approx 0.129$) because the bound counts only trajectories trapped on arm $k$, while runs locked onto the other suboptimal arms also contribute regret.
Figure~\ref{fig:trap} (right) validates the theoretical recursion $\rho_k$ across prior configurations (e.g., $0.787 \pm 0.006$ empirical vs.\ $0.790$ theory at $\beta_{\min}^{(-k)}=5$). Although Theorem~\ref{thm:lock-in} states $\beta_{\min}^{(-k)}>1$ as a sufficient condition, we observe an empirical boundary at $\beta_{\min}^{(-k)}=1$. Below this threshold ($\beta \le 1$), prior mass near $1$ allows competing arms to occasionally generate high Thompson samples and break the trap. Finally, hollow triangles confirm that perpetual selection $\{a_t{=}k \ \forall t\}$ strictly bounds the theoretical lock-in event $\mathcal{E}_k$.

\subsection{Actions Alone Cannot Reveal User Bias}
\label{sec:exp-nonid}

\textbf{Takeaway 2 (Bias is invisible)}: Gradient-Boosted Trees, Logistic Regression, and the Bayes-optimal rule all classify agent bias at random accuracy ($0.500 \pm 0.003$) from action histories alone. The action-only posterior remains frozen at its prior.

We test Theorem~\ref{thm:non-identifiability} in Figure~\ref{fig:passive-impossibility} (extended trajectories in Appendix~\ref{sec:app-nonid}).
As shown in Figure~\ref{fig:passive-impossibility} (left), non-linear classifiers (GBTs, Logistic Regression) and Monte Carlo estimates of the Bayes-optimal decision rule achieve empirical classification accuracies statistically indistinguishable from $0.500$ across all horizons $T \in [25, 800]$. Extended evaluations across six parameter configurations (Table~\ref{tab:nonid-appendix}) confirm total variation bounds $\mathrm{TV}(Q_0, Q_1) \le 0.016$, matching the finite-sample noise floor $\mathrm{TV}_0$.
Figure~\ref{fig:passive-impossibility} (right) illustrates the frozen mechanism: while a reward-aware observer updates its belief as environmental payoffs are revealed, an observer restricted strictly to action histories $h_t^a$ experiences zero information gain, $P(M_1 \mid h_t^a) \equiv \pi_1 = 0.5$ for all $t$. Thus, passive observation cannot separate bias from environmental confounding, rendering active interventions necessary.

\begin{figure}[t]
\vspace{-0.5cm}
\centering
\begin{minipage}{0.49\linewidth}
\centering
\includegraphics[width=\linewidth]{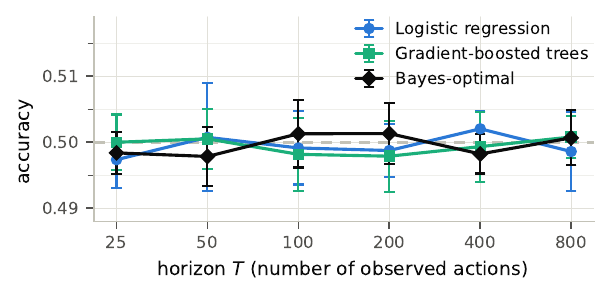}
\end{minipage}\hfill
\begin{minipage}{0.49\linewidth}
\centering
\includegraphics[width=\linewidth]{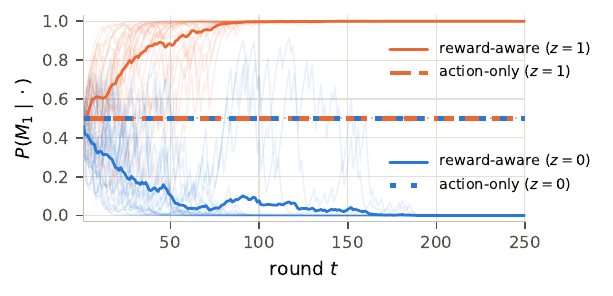}
\end{minipage}
\caption{\textbf{Passive observation fails to identify user bias (Theorem~\ref{thm:non-identifiability}).}  \\
\textbf{Left:} Out-of-sample classification accuracy for user type $z \in \{0, 1\}$ using action histories $h_T^a$ up to $T=800$ ($10{,}000$ trajectories/class, $5$-fold CV, $95\%$ CIs). All decision models stay at random level.  \\
\textbf{Right:} Posterior trajectory $P(M_1 \mid h_t^a)$ for an action-only observer (thick dashed line, frozen at $\pi_1=0.5$) versus a reward-aware oracle (thin lines, separating toward $0$ or $1$, bold is average).}
\label{fig:passive-impossibility}
\end{figure}

\subsection{Active Interventions Restore Identifiability}
\label{sec:exp-active}

\textbf{Takeaway 3 (Testing is cheap, not testing is not)}: Both active policies reach $90\%$ separation within $16$–$27$ interaction rounds, an order of magnitude faster than the worst-case bound $m_\delta=260$. Crucially, the non-coercive policy imposes $40\%$ fewer collateral forced failures on unbiased users.

\begin{figure}[htbp]
\centering
\begin{minipage}{0.49\linewidth}
\centering
\includegraphics[width=\linewidth]{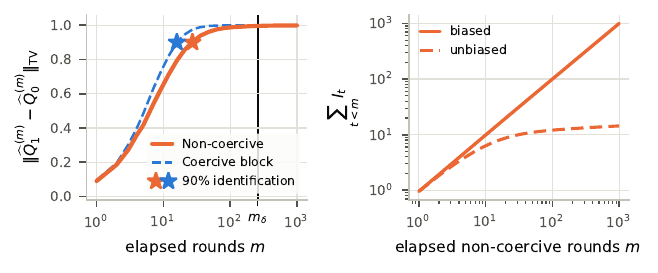}
\end{minipage}\hfill
\begin{minipage}{0.47\linewidth}
\centering
\includegraphics[width=\linewidth]{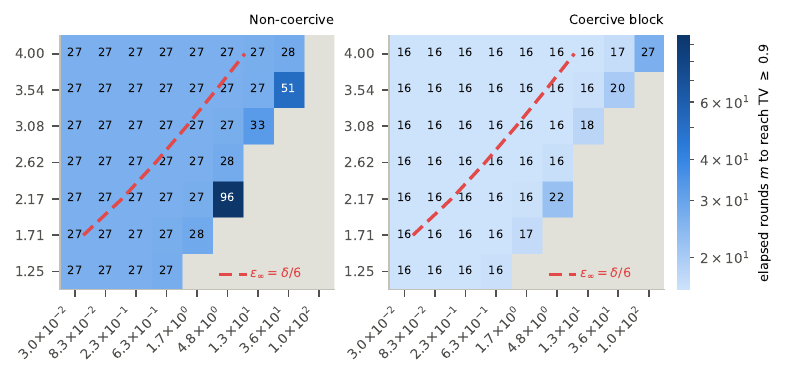}
\end{minipage}
\caption{\textbf{Active identification performance and feasibility phase diagram.} \\
\textbf{Left:} Total Variation separation $\Vert \widehat{Q}_1^{(m)} - \widehat{Q}_0^{(m)} \Vert_{\mathrm{TV}}$ vs.\ forced failures borne by an unbiased user under coercive vs.\ non-coercive policies ($20{,}000$ runs, diagnostic prior $\mathrm{Beta}(5,1)$, competing arms at $\mathrm{Beta}(2,4)$ and $\mathrm{Beta}(3,5)$, kernel $(\gamma,C,p)=(0.5,0.01,2)$, $\delta=0.1$). \\
\textbf{Right:} Elapsed rounds to reach $\mathrm{TV} \ge 0.9$ across slack $C$ and exponent $p$ for non-coercive (top) vs.\ coercive (bottom) policies ($8{,}000$ seeds/cell; grey cells indicate non-reachability within $600$ rounds). Dashed line marks Theorem~\ref{thm:identifiability}'s feasibility boundary $\varepsilon_\infty \le \delta/6$.}
\label{fig:headline}
\end{figure}

Theorem~\ref{thm:identifiability} shows that coercive zero-reward blocks (Def.~\ref{def:interventions}) identify a belief-dependent bias within a deterministic bound $m_\delta$, while Prop.~\ref{appendix:sequantial-id} extends to non-coercive settings when spontaneous diagnostic selection occurs. As shown in Figure~\ref{fig:headline} (left), active interventions achieve distribution separation ($\mathrm{TV} \to 1$). Target separation ($\mathrm{TV} \ge 0.9$) requires only $16$ rounds (coercive) and $27$ rounds (non-coercive), well below $m_\delta = 260$. By halting interventions once unbiased agents update beliefs, the non-coercive policy caps expected unbiased user cost at $9.6$ forced failures, a $40\%$ reduction over coercive execution ($16.1$). Figure~\ref{fig:headline} (right) maps identification across parameters $(C, p)$. Within the certified boundary, separation times remain invariant at $16$--$27$ rounds. Empirically, active testing succeeds up to two orders of magnitude beyond theoretical bounds before reaching non-reachability ($>600$ rounds). As bias weakens, action distributions merge, making coercive blocks necessary for broader parameter coverage.
\section{Discussion}

Our results isolate a structural failure mode for alignment from behavior: when actions are jointly shaped by an unknown environment and a distortion of experienced evidence, passive observations can not reveal which mechanism generated them. Within our model, active interventions can break this limitation and turn uninformative action traces into signal. 
We establish this under controlled assumptions, not a deployment-ready test of human cognition.

\paragraph{Scope of the formal claims.}
The paper uses different bias models for different questions (Table~\ref{tab:scope}). Exact passive non-identifiability is proved for a memoryless static channel, whereas the confirmation trap and active identification concern a belief-dependent kernel. This split is informative rather than incidental: zero-reward interventions exploit persistent confirmation and can be uninformative for a static channel (Remark~\ref{rem:static-blind}). Two cells remain open: whether passive observation can identify a belief-dependent kernel, and which interventions identify a static channel with $v>0$. Human choices may instead reflect graded and time-varying bias, risk preferences, habits, strategic behavior, changing goals, or simple model misspecification. Our identification guarantees hold relative to the assumed model class. Future work should extend this framework to continuous and heterogeneous types, dynamic kernels or alternative policies.

\paragraph{From separability to useful intervention.}
The computable horizon $m_\delta$ is a conservative upper bound. While empirical separation requires far fewer rounds ($16$--$27$), even this lower range is far too disruptive for real-world interactions, where forcing $16$ consecutive failures would cause severe user frustration, loss of trust, or abandonment. Furthermore, identifying a bias does not imply correcting it or improving user welfare. Practical protocols must therefore adapt intervention strength to real-time uncertainty, stop as soon as information gain no longer justifies the user burden, and measure downstream welfare rather than separability alone.

\paragraph{Empirical and human validity.}
Our simulations confirm our analytical bounds under set assumptions, but they do not validate these behavioral dynamics in human users. When an LLM acts as the observer, it extracts signals to detect user bias with interventions, but the gap between a peak AUROC of $0.92$ and a lower accuracy highlights a critical calibration bottleneck. Additionally, we show that LLMs have some limitations mirroring self-confirming bias. Validating distortion kernels and calibrating stopping rules across heterogeneous users require joint studies with behavioral scientists.

\paragraph{Intervention design and ethics.}
While forcing zero rewards is analytically clean, it risks manipulating user experience. Non-coercive protocols reduce forced failures, but they do not resolve core issues of consent, transparency, or false-positive harms. Fabricating outcomes should be strictly prohibited. Safe evaluations must rely on consented, low-stakes environments with verifiable counter-evidence. These frameworks must account for cases where user reflects a correct belief rather than a harmful cognitive bias, requiring ground-truth benchmarks. Finally, randomized controls and calibrated interventions are necessary to separate belief updates from user fatigue, annoyance, or distrust.

The broader lesson is that behavioral agreement does not prove genuine preference, and no volumes of logs repairs a structural confound.
Active elicitation can provide missing signal, but only if its calibration, human costs, and governance are treated as core scientific questions rather than implementation details.

\begin{ack}
    Tianyi Qiu and Zhonghao He are supported by the Cosmos Fellowship, affiliated with the Oxford HAI Lab while working on this project.
\end{ack}


{
\small

\bibliographystyle{unsrt}
\bibliography{references} 

}

\newpage
\appendix

\section{Proofs}

\subsection{Proof of \ref{thm:lock-in}}\label{proof:lock-in}

Let
\[
\mathcal E_{k,n}:=\bigcap_{t=1}^{n}\{a_t=k,\ \tilde r_t=1\},\qquad \mathcal E_{k,0}:=\Omega
\]
The events form a decreasing sequence and satisfy
\[
\mathcal E_k=\bigcap_{n=0}^{\infty}\mathcal E_{k,n}
\]
Conditional on $\mathcal E_{k,n}$, arm $k$ has been selected $n$ times and all its subjectively perceived rewards have been equal to one. By Beta--Bernoulli conjugacy, its subjective posterior is therefore
\[
\mu_k\mid\mathcal E_{k,n}\sim\operatorname{Beta}(\alpha_k+n,\beta_k)
\]
Because no competing arm has been selected on $\mathcal E_{k,n}$, the posterior of every arm $j\neq k$ remains equal to its prior:
\[
\mu_j\mid\mathcal E_{k,n}\sim\operatorname{Beta}(\alpha_j,\beta_j),\qquad j\neq k
\]
In particular, the posterior mean of arm $k$ on this history is
\[
m_n:=\tilde\mu_k^{(n)}=\frac{\alpha_k+n}{\alpha_k+\beta_k+n}
\]
and hence
\[
1-m_n=\frac{\beta_k}{\alpha_k+\beta_k+n}
\]

We first bound the probability of a contradictory subjective perception. Define
\[
u_n:=P(\tilde r_{n+1}=0\mid\mathcal E_{k,n},a_{n+1}=k)
\]
Conditional on selecting arm $k$, the objective reward is distributed as $\operatorname{Bernoulli}(\mu_k)$. Therefore,
\[
u_n=\mu_kP_{\mathrm{bias}}(\tilde r=0\mid r=1,m_n)+(1-\mu_k)P_{\mathrm{bias}}(\tilde r=0\mid r=0,m_n)
\]
By Assumption 2, the first term is zero and the second term is bounded by $C(1-m_n)^p$. Thus,
\[
u_n\leq(1-\mu_k)C(1-m_n)^p=(1-\mu_k)C\left(\frac{\beta_k}{\alpha_k+\beta_k+n}\right)^p
\]
Consequently, there exists a finite constant $A_k>0$ such that
\[
u_n\leq A_k(n+1)^{-p}
\]
Since $p>1$, it follows that
\[
\sum_{n=0}^{\infty}u_n<\infty
\]
Moreover, every factor $1-u_n$ is strictly positive. Indeed, if $r=1$, Assumption 2 implies $\tilde r=1$ almost surely. If $r=0$, then $m_n>0$ and Assumption 1 gives
\[
P_{\mathrm{bias}}(\tilde r=1\mid r=0,m_n)=\mathbb E[\tilde r\mid r=0,m_n]\geq\gamma m_n>0
\]
It follows that
\[
1-u_n\geq\mu_k+(1-\mu_k)\gamma m_n>0
\]

We next bound the probability that Thompson Sampling abandons arm $k$. Let
\[
X_n\sim\operatorname{Beta}(\alpha_k+n,\beta_k)
\]
denote the Thompson sample for arm $k$, and let
\[
Y_j\sim\operatorname{Beta}(\alpha_j,\beta_j),\qquad j\neq k
\]
denote the independent Thompson samples of the competing arms. Define the conditional escape probability
\[
q_n:=P(a_{n+1}\neq k\mid\mathcal E_{k,n})=P\left(\max_{j\neq k}Y_j>X_n\right)
\]
For each $j\neq k$, there exists a finite constant $D_j>0$ such that the upper tail of the Beta distribution satisfies
\[
P(Y_j>x)\leq D_j(1-x)^{\beta_j},\qquad x\in[0,1]
\]
To verify this bound, for $x\geq1/2$ we have
\[
P(Y_j>x)=\frac{1}{B(\alpha_j,\beta_j)}\int_x^1y^{\alpha_j-1}(1-y)^{\beta_j-1}\,dy\leq\frac{M_j}{\beta_jB(\alpha_j,\beta_j)}(1-x)^{\beta_j}
\]
where $M_j:=\sup_{y\in[1/2,1]}y^{\alpha_j-1}<\infty$. For $x<1/2$, the same inequality holds after increasing the constant, since
\[
P(Y_j>x)\leq1\leq2^{\beta_j}(1-x)^{\beta_j}
\]
Using the union bound and conditioning on $X_n$, we obtain
\[
q_n\leq\sum_{j\neq k}P(Y_j>X_n)\leq\sum_{j\neq k}D_j\mathbb E[(1-X_n)^{\beta_j}]
\]
The relevant Beta moment is
\[
\mathbb E[(1-X_n)^{\beta_j}]=\frac{B(\alpha_k+n,\beta_k+\beta_j)}{B(\alpha_k+n,\beta_k)}
\]
Using the definition of the Beta function,
\[
\mathbb E[(1-X_n)^{\beta_j}]=\frac{\Gamma(\beta_k+\beta_j)}{\Gamma(\beta_k)}\frac{\Gamma(\alpha_k+\beta_k+n)}{\Gamma(\alpha_k+\beta_k+\beta_j+n)}
\]
The standard Gamma-ratio asymptotic gives
\[
\frac{\Gamma(\alpha_k+\beta_k+n)}{\Gamma(\alpha_k+\beta_k+\beta_j+n)}=O(n^{-\beta_j})
\]
Hence, for each $j\neq k$, there exists a finite constant $G_j>0$ such that
\[
\mathbb E[(1-X_n)^{\beta_j}]\leq G_j(n+1)^{-\beta_j}
\]
It follows that there exists a finite constant $G>0$ such that
\[
q_n\leq G(n+1)^{-\beta_{\min}^{(-k)}}
\]
Because $\beta_{\min}^{(-k)}>1$, we obtain
\[
\sum_{n=0}^{\infty}q_n<\infty
\]
Furthermore, $q_n<1$ for every finite $n$, because all proper Beta distributions have positive density on $(0,1)$ and hence
\[
P\left(X_n>\max_{j\neq k}Y_j\right)>0
\]

By the chain rule of conditional probability,
\[
P(\mathcal E_{k,n+1}\mid\mathcal E_{k,n})=(1-q_n)(1-u_n)
\]
Therefore, for every $N\geq1$,
\[
P(\mathcal E_{k,N})=\prod_{n=0}^{N-1}(1-q_n)(1-u_n)
\]
Since $\sum_nq_n<\infty$ and $\sum_nu_n<\infty$, there exists $N_0$ such that $q_n\leq1/2$ and $u_n\leq1/2$ for all $n\geq N_0$. Using $\log(1-x)\geq-2x$ for $x\in[0,1/2]$, we have
\[
\sum_{n=N_0}^{\infty}\log\bigl((1-q_n)(1-u_n)\bigr)\geq-2\sum_{n=N_0}^{\infty}(q_n+u_n)>-\infty
\]
The finitely many factors preceding $N_0$ are strictly positive. Consequently,
\[
\prod_{n=0}^{\infty}(1-q_n)(1-u_n)>0
\]
By continuity of probability for decreasing events,
\[
\rho_k:=P(\mathcal E_k)=\lim_{N\to\infty}P(\mathcal E_{k,N})=\prod_{n=0}^{\infty}(1-q_n)(1-u_n)>0
\]
Since $\mathcal E_k\subseteq\bigcap_{t=1}^{\infty}\{a_t=k\}$, this proves
\[
P(a_t=k\text{ for every }t\geq1)\geq\rho_k>0
\]

On $\mathcal E_k$, the posterior mean of arm $k$ satisfies
\[
\tilde\mu_k^{(n)}=\frac{\alpha_k+n}{\alpha_k+\beta_k+n}\longrightarrow1
\]
Moreover,
\[
P(a_{n+1}=k\mid\mathcal E_{k,n})=1-q_n\longrightarrow1
\]
Thus, the subjective evidence becomes increasingly favorable to arm $k$, while the probability of exploring another arm vanishes.

Finally, define the pseudo-regret by
\[
R_T:=\sum_{t=1}^T(\mu^\star-\mu_{a_t})
\]
On $\mathcal E_k$, the agent selects arm $k$ at every time step, so
\[
R_T=\Delta_kT
\]
Since pseudo-regret is nonnegative outside $\mathcal E_k$,
\[
\mathbb E[R_T]\geq P(\mathcal E_k)\Delta_kT=\rho_k\Delta_kT
\]
Therefore, the confirmation trap occurs with strictly positive probability and produces linear expected pseudo-regret. \hfill $\square$

\subsection{Proof of \ref{thm:non-identifiability}}\label{proof:non-identifiability}

Let
\[
G_t^1:=\big((a_s,\tilde r_s)\big)_{s=1}^{t-1}
\]
denote the recorded history of the biased agent, and let
\[
G_t^0:=\big((a_s,r_s^0)\big)_{s=1}^{t-1}
\]
denote the objective history of the unbiased agent, where $r_s^0$ denotes the reward observed under the effective environment $\boldsymbol{\nu}$. Histories are ordered sequences, and $g\oplus(k,y)$ denotes the history $g$ extended by the pair $(k,y)$.

Throughout, the objective reward $r_t$ is drawn from $\operatorname{Bernoulli}(\mu_{a_t})$ (resp.\ $\operatorname{Bernoulli}(\nu_{a_t})$ under $M_0$) independently of the past given $a_t$, and the perception noise at round $t$ is independent of the past and of $a_t$ given $r_t$. The posterior is a fixed function of the recorded history, common to both agents, with a common convention when the prior predictive probability of a history vanishes (this never happens with Beta priors). Thompson draws are fresh at each round and independent of rewards and perception noise, and ties are broken by a common rule. Hence the conditional law of $a_t$ given everything that happened before round $t$ depends only on the recorded history, through a kernel $\pi_{\mathrm{TS}}(k\mid g)$ that is the same under both models (the selection probability of Section~\ref{sec:setup} evaluated at the posterior computed from $g$). This is also why the observer's action likelihood in Section~\ref{sec:setup} is an expectation: the belief $\tilde b_{t-1}$ depends on rewards that $\mathcal{A}_{\text{obs}}$ does not see, so it is not a function of $h_t^a$.

Consider the biased model $M_1$. Conditional on any recorded history $G_t^1=g_t$ and on selecting arm $k$, the objective reward remains distributed according to the stationary Bernoulli environment:
\[
P_{M_1}(r_t=1\mid G_t^1=g_t,a_t=k)=\mu_k
\]
Using the memoryless property of the perceptual channel and the law of total probability, we obtain
\[
P_{M_1}(\tilde r_t=1\mid G_t^1=g_t,a_t=k)=u\mu_k+v(1-\mu_k)=\nu_k
\]
Therefore,
\[
\tilde r_t\mid(G_t^1=g_t,a_t=k)\sim\operatorname{Bernoulli}(\nu_k)
\]

Under the unbiased model $M_0$, arm $k$ has objective mean $\nu_k$. Hence,
\[
r_t^0\mid(G_t^0=g_t,a_t=k)\sim\operatorname{Bernoulli}(\nu_k)
\]
Thus, conditional on the same recorded history and selected action, the next recorded observation has the same distribution under both models:
\[
P_{M_1}(\tilde r_t=y\mid G_t^1=g_t,a_t=k)=P_{M_0}(r_t^0=y\mid G_t^0=g_t,a_t=k)
\]
for every $y\in\{0,1\}$.

We now prove by induction that the recorded-history processes $(G_t^1)_{t\geq1}$ and $(G_t^0)_{t\geq1}$ have the same distribution. At $t=1$, both histories are empty, so $G_1^1$ and $G_1^0$ trivially have the same distribution; the induction step below then covers the first action $a_1$.

Suppose that $G_t^1$ and $G_t^0$ have the same distribution for some $t\geq1$. Conditional on any common realization
\[
g_t=\big((k_s,y_s)\big)_{s=1}^{t-1}
\]
with positive probability, both agents have updated the same initial prior using the same sequence of recorded action--reward pairs and the same Bernoulli likelihood. Their posterior distributions are therefore identical.

Because both agents apply the same Thompson Sampling rule to these identical posteriors, their conditional action probabilities coincide:
\[
P_{M_1}(a_t=k\mid G_t^1=g_t)=P_{M_0}(a_t=k\mid G_t^0=g_t)=\pi_{\mathrm{TS}}(k\mid g_t)
\]
for every arm $k$.

Conditional on the selected action, the distributions of their next recorded observations also coincide:
\[
P_{M_1}(\tilde r_t=y\mid G_t^1=g_t,a_t=k)=P_{M_0}(r_t^0=y\mid G_t^0=g_t,a_t=k)
\]
Consequently, for every possible extension $(k,y)$ of $g_t$ and $z\in\{0,1\}$,
\[
P_{M_z}\big(G_{t+1}^z=g_t\oplus(k,y)\big)=P_{M_z}(G_t^z=g_t)\,\pi_{\mathrm{TS}}(k\mid g_t)\,\nu_k^{\,y}(1-\nu_k)^{1-y}
\]
where every factor on the right-hand side is the same under both models. Hence
\[
P_{M_1}(G_{t+1}^1=g_t\oplus(k,y))=P_{M_0}(G_{t+1}^0=g_t\oplus(k,y))
\]
It follows that $G_{t+1}^1$ and $G_{t+1}^0$ have the same distribution. By induction, the joint distributions of the recorded histories coincide at every horizon.

Marginalizing out the recorded rewards yields
\[
P_{M_1}^{a,T}=P_{M_0}^{a,T}
\]
for every $T\geq1$. Hence, the biased agent in environment $\boldsymbol{\mu}$ and the unbiased agent in environment $\boldsymbol{\nu}$ are observationally equivalent from the perspective of an action-only observer.

It remains to establish the statistical lower bound. Let
\[
Q_T:=P_{M_0}^{a,T}=P_{M_1}^{a,T}
\]
denote their common action-history distribution, and define
\[
q_T:=Q_T(\widehat z_T=1)
\]
Under the unbiased model, the estimator's probability of error is
\[
P_{M_0}(\widehat z_T=1)=q_T
\]
Under the biased model, its probability of error is
\[
P_{M_1}(\widehat z_T=0)=1-q_T
\]
Therefore,
\[
P_{M_0}(\widehat z_T=1)+P_{M_1}(\widehat z_T=0)=1
\]
Consequently,
\[
\max\left\{P_{M_0}(\widehat z_T=1),P_{M_1}(\widehat z_T=0)\right\}\geq\frac12
\]
This bound holds for every estimator and every horizon $T$. Therefore, no sequence of action-only estimators can consistently distinguish the two models, and no action-only estimator can be uniformly consistent over model classes containing both models. The same bound holds for randomized estimators whose internal randomness is independent of the interaction, by conditioning on that randomness.

Finally, we consider the observer's posterior odds. By the observational equivalence established above, the two models assign the same likelihood to every action history $h_{T+1}^a=(a_1,\dots,a_T)$:
\[
P(h_{T+1}^a\mid M_0)=P(h_{T+1}^a\mid M_1)=:L(h_{T+1}^a)
\]
For every action history with $L(h_{T+1}^a)>0$ (hence almost surely), and with $\pi_0=1-\pi_1$, Bayes' rule therefore gives
\[
P(M_1\mid h_{T+1}^a)=\frac{\pi_1L(h_{T+1}^a)}{\pi_0L(h_{T+1}^a)+\pi_1L(h_{T+1}^a)}=\frac{\pi_1}{\pi_0+\pi_1}=\pi_1
\]
Thus, observing additional actions cannot update the observer's posterior odds between the two observationally equivalent models. \hfill $\square$

\subsection{Proof of Remark~\ref{rem:static-blind}}\label{proof:static-blind}

Eq.~\ref{eq:confirmation} requires $P_{\mathrm{bias}}(\tilde r=1\mid r=1)=1$, i.e.\ $u=1$, and $1-v \le C(1-\tilde\mu)^p$ for $\tilde\mu$ arbitrarily close to $1$, i.e.\ $v=1$. A static channel therefore satisfies Persistent Confirmation only if $u=v=1$.

Now let $v=0$. Consider any intervention policy, coercive or not, in which the forced action (if any) at round $t$ is chosen from $(a_{1:t-1}, I_{1:t-1})$ and $I_t$ from $(a_{1:t}, I_{1:t-1})$, possibly with randomization independent of rewards, perception noise and Thompson draws. Both users record forced rounds and zeroed rewards as ordinary observations, as in Theorem~\ref{thm:identifiability}. We show that the biased user in $\boldsymbol{\mu}$ and the unbiased user in $\boldsymbol{\nu}$ induce identical joint laws of $(a_{1:T}, I_{1:T})$ for every $T$. The induction of Appendix~\ref{proof:non-identifiability} goes through with the state augmented by the intervention history. Conditionally on the recorded history, $I_{1:t-1}$, $a_t$ and $I_t$, the recorded outcome is $0$ under both models when $I_t=1$ ($r_t=0$ and $v=0$ give $\tilde r_t=0$), and $\operatorname{Bernoulli}(\nu_{a_t})$ under both models when $I_t=0$. The forcing and intervention kernels depend only on the observer's information, so they are the same under both models. Consequently, every estimator based on actions and interventions is at chance level and $P(M_1 \mid a_{1:T}, I_{1:T}) = \pi_1$ almost surely.

For $v>0$, intervened rounds are perceived as $\operatorname{Bernoulli}(v)$ under $M_1$ but recorded as $0$ under $M_0$, so the confound breaks, but Theorem~\ref{thm:identifiability} does not quantify the resulting separation. By contrast, a user satisfying Persistent Confirmation whose confidence in arm $k$ is already high ($\varepsilon_\infty(A,B)\le\delta/6$) perceives zeroed rewards as successes with probability at least $1-C(1-\tilde\mu)^p$, which tends to one along confirming paths; this is what the block of Theorem~\ref{thm:identifiability} exposes. \hfill $\square$

\subsection{Proof of \ref{thm:identifiability}}\label{proof:identifiability}

Let
\[
E_m:=\bigcap_{\ell=0}^{m-1}\{\tilde r_{\tau+\ell}=1\}
\]
denote the event that the biased user transforms every injected objective failure into a perceived success.

Conditional on $E_\ell$, the biased posterior of arm $k$ after $\ell$ interventions is
\[
\operatorname{Beta}(A+\ell,B)
\]
Its posterior mean is
\[
\tilde\mu_{k,\ell}=\frac{A+\ell}{A+B+\ell}
\]
and hence
\[
1-\tilde\mu_{k,\ell}=\frac{B}{A+B+\ell}
\]
By Assumption 2,
\[
P_{M_1}\left(\tilde r_{\tau+\ell}=0\mid E_\ell\right)\leq C\left(\frac{B}{A+B+\ell}\right)^p
\]
Decomposing according to the first perceived failure and applying a union bound yields
\[
P_{M_1}(E_m^c)\leq CB^p\sum_{\ell=0}^{m-1}(A+B+\ell)^{-p}
\]
Therefore,
\[
P_{M_1}(E_m^c)\leq\varepsilon_\infty(A,B)
\]

On $E_m$, the biased posterior of arm $k$ is
\[
\operatorname{Beta}(A+m,B)
\]
and its posterior mean is
\[
\frac{A+m}{A+B+m}\longrightarrow1
\]

Under $M_0$, every injected reward is observed as an objective failure. The posterior of arm $k$ is therefore
\[
\operatorname{Beta}(A,B+m)
\]
and its posterior mean is
\[
\frac{A}{A+B+m}\longrightarrow0
\]

Let
\[
X_m\sim\operatorname{Beta}(A,B+m)
\]
For a fixed threshold $x\in(0,1)$,
\[
P(X_m>x)=\frac{1}{B(A,B+m)}\int_x^1s^{A-1}(1-s)^{B+m-1}ds
\]
Since
\[
s^{A-1}\leq D_A(x)
\]
on $[x,1]$, it follows that
\[
P(X_m>x)\leq D_A(x)\frac{\Gamma(A+B+m)}{\Gamma(A)\Gamma(B+m+1)}(1-x)^{B+m}
\]
Thus,
\[
P(X_m>x)\leq L_0(m;x)
\]
The Gamma-function ratio grows only polynomially in $m$, whereas $(1-x)^m$ decreases exponentially. Hence,
\[
L_0(m;x)\longrightarrow0
\]

Now let
\[
X_m'\sim\operatorname{Beta}(A+m,B)
\]
Since $1-X_m'$ has distribution $\operatorname{Beta}(B,A+m)$, the same calculation gives
\[
P(X_m'\leq1-y)\leq L_1(m;y)
\]
and
\[
L_1(m;y)\longrightarrow0
\]

Under $M_0$, if arm $k$ is selected at the diagnostic round, then either $X_m>x$ or the Thompson draw of competitor $j_0$ is at most $x$. Therefore,
\[
Q_0^{(m)}(k)\leq L_0(m;x)+G_{j_0}(x)
\]

Conditional on $E_m$, arm $k$ is selected under $M_1$ whenever $X_m'>1-y$ and every competitor draw is below $1-y$. Consequently,
\[
Q_1^{(m)}(k\mid E_m)\geq1-L_1(m;y)-\sum_{j\neq k}\left[1-G_j(1-y)\right]
\]
After marginalizing over the latent perceived rewards,
\[
Q_1^{(m)}(k)\geq1-\varepsilon_\infty(A,B)-L_1(m;y)-\sum_{j\neq k}\left[1-G_j(1-y)\right]
\]

By construction,
\[
\varepsilon_\infty(A,B)\leq\xi
\]
\[
G_{j_0}(x)\leq\xi
\]
\[
\sum_{j\neq k}\left[1-G_j(1-y)\right]\leq\xi
\]
For every $m\geq m_\delta$,
\[
L_0(m;x)\leq\xi
\]
and
\[
L_1(m;y)\leq\xi
\]
Therefore,
\[
Q_1^{(m)}(k)-Q_0^{(m)}(k)\geq1-5\xi
\]
Since $\xi=\delta/6$,
\[
1-5\xi=1-\frac{5\delta}{6}>1-\delta
\]
Finally, total variation dominates the probability difference of every measurable event. In particular,
\[
\left\|Q_1^{(m)}-Q_0^{(m)}\right\|_{\mathrm{TV}}\geq\left|Q_1^{(m)}(k)-Q_0^{(m)}(k)\right|>1-\delta
\]
\hfill $\square$
\section{Reward-only sequential identification}\label{appendix:sequantial-id}
    Suppose that the Observer cannot force actions and uses the policy
    \[
    I_s=\mathbf 1\{a_s=k\}
    \]
    during a diagnostic window beginning at time $\tau$. Define the no-switch event
    \[
    S_m:=\bigcap_{\ell=0}^{m-1}\{a_{\tau+\ell}=k\}
    \]
    Let
    \[
    q_{1,\ell}(k):=q_k\left(\operatorname{Beta}(A+\ell,B),b_{-k}\right)
    \]
    denote the Thompson Sampling probability of choosing arm $k$ after $\ell$ perceived successes, with the competitor posteriors held fixed. If
    \[
    \varepsilon_\infty(A,B)+\sum_{\ell=0}^{\infty}\left[1-q_{1,\ell}(k)\right]\leq\varepsilon
    \]
    then
    \[
    P_{M_1}(S_m)\geq1-\varepsilon
    \]
    Under $M_0$, define
    \[
    q_{0,\ell}(k):=q_k\left(\operatorname{Beta}(A,B+\ell),b_{-k}\right)
    \]
    Then
    \[
    P_{M_0}(S_m)=\prod_{\ell=0}^{m-1}q_{0,\ell}(k)\longrightarrow0
    \]
    Consequently,
    \[
    \left\|P_{M_1}^{a,m}-P_{M_0}^{a,m}\right\|_{\mathrm{TV}}\geq P_{M_1}(S_m)-P_{M_0}(S_m)
    \]
    and the action-path distributions become arbitrarily well separated whenever the confirming-path escape probabilities are summable.

\begin{proof}
    The result follows from a sequential union bound over perceived failures and Thompson Sampling deviations under $M_1$. Under $M_0$, conditional on the no-switch path, arm $k$ accumulates deterministic failures, giving the displayed product representation. The factors $q_{0,\ell}(k)$ converge to zero, so the product converges to zero.
\end{proof}
\section{Additional Experimental Results}
\label{sec:exp-appendix}

This appendix presents supplementary simulation figures, full parameter sweeps, and diagnostic tests referenced in the main text (\S\ref{sec:exp-appendix-supp}), followed by empirical grounding experiments evaluating Large Language Models (\S\ref{sec:exp-llm}).

\begin{table}[htbp]
\centering
\small
\renewcommand{\arraystretch}{1.4}
\caption{\textbf{Overview of experimental findings.} Key quantitative results and where to find them.}
\label{tab:overview}
\makebox[\textwidth][c]{%
\begin{tabular}{p{0.30\textwidth} p{0.10\textwidth} p{0.39\textwidth} c}
\toprule
\textbf{Research Question} & \textbf{Outcome} & \textbf{Key Result / Metric} & \textbf{Ref.} \\
\midrule
How frequent and costly are confirmation traps? & High Cost & Regret increases $\times 17$ ($R_t/t \approx 0.3$ vs.\ $0$); $\hat{\rho}_k$ matches theory within $\pm 0.008$ & \S\ref{sec:exp-trap} \\
Can passive observation distinguish biased agents? & No & Accuracy $\approx 50\%$ (across GBTs, LogReg, Bayes-optimal, and LLMs) & \S\ref{sec:exp-nonid}, \S\ref{sec:app-nonid} \\
Do active interventions separate bias efficiently? & Yes & Target separation ($\text{TV} \ge 0.9$) in 16–27 rounds vs.\ worst-case bound $m_\delta = 260$ & \S\ref{sec:exp-active} \\
Do theoretical signatures transfer to LLMs? & Yes & LLM observers reach $0.92$ AUROC under active logs & App.~\ref{sec:exp-llm} \\
\bottomrule
\end{tabular}
}
\end{table}

\subsection{Supplementary Simulation Figures and Tables}
\label{sec:exp-appendix-supp}

\subsubsection{Extended Non-Identifiability Sweeps and Action Overlay}
\label{sec:app-nonid}

In Section~\ref{sec:exp-nonid}, Theorem~\ref{thm:non-identifiability} establishes that passive observation over action histories yields zero information gain regarding user bias $z \in \{0, 1\}$. Figure~\ref{fig:action-overlay} and Table~\ref{tab:nonid-appendix} provide extended empirical support for this equivalence across various bandit configurations.

\begin{figure}[htbp]
\centering
\includegraphics[width=0.85\linewidth]{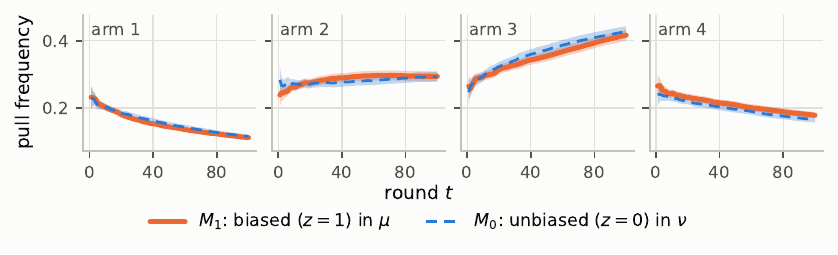}
\caption{\textbf{Action trajectory superposition under Static Binary Bias.} Per-arm selection frequencies for a biased agent $M_1$ under Static Binary Bias ($(u,v)=(0.75,0.15)$ in environment $\boldsymbol\mu=(0.20,0.50,0.60,0.35)$) versus an unbiased agent $M_0$ in confounded environment $\boldsymbol\nu=(0.27,0.45,0.51,0.36)$ ($\mathrm{Beta}(1,1)$ priors, $600$ seeds per model, $95\%$ confidence bands, $T=100$). As predicted by Theorem~\ref{thm:non-identifiability}, action marginals overlap across all time steps.}
\label{fig:action-overlay}
\end{figure}

\paragraph{Action Trajectory Overlap.}
Figure~\ref{fig:action-overlay} visualizes temporal arm selection probabilities for $M_1$ (biased) and $M_0$ (unbiased under confounded rewards $\boldsymbol{\nu}$). Across all four arms, empirical action frequencies match throughout the $T=100$ horizon. The maximum gap in pull frequency at any single round remains within $0.045$ units ($\le 2.44$ standard errors), demonstrating that action distributions are empirically indistinguishable.

\begin{table}[htbp]
\centering
\small
\caption{\textbf{Non-identifiability evaluation across six parameter configurations} at horizon $T=500$ ($20{,}000$ simulation trajectories per configuration). Total Variation ($\mathrm{TV}$) is estimated via $10$-bin summary histogram representations of action traces and compared against the finite-sample noise floor ($\mathrm{TV}_0$). Classification accuracy ($\mathrm{Acc}_{\text{clf}}$) is evaluated using Gradient Boosted Trees ($5$-fold CV). $\mathrm{Acc}_{\text{Bayes}}$ denotes Monte Carlo estimates of the Bayes-optimal decision boundary. Unadjusted permutation test $p$-values are reported for completeness.}
\label{tab:nonid-appendix}
\begin{tabular}{lcc ccc cc}
\toprule
\textbf{Configuration} & $K$ & $(u,v)$ & \textbf{TV} & $\mathbf{TV_0}$ & $p$ & $\mathbf{Acc}_{\text{clf}}$ \textbf{(95\% CI)} & $\mathbf{Acc}_{\text{Bayes}}$ \\
\midrule
base            & 4 & (0.75, 0.15) & 0.0142 & 0.0111 & 0.040 & 0.500\ [0.498, 0.503] & 0.501 \\
sharp-gap       & 3 & (0.70, 0.10) & 0.0087 & 0.0071 & 0.159 & 0.502\ [0.499, 0.505] & 0.499 \\
order-reversing & 3 & (0.30, 0.80) & 0.0073 & 0.0085 & 0.716 & 0.500\ [0.493, 0.506] & 0.500 \\
near-identity   & 2 & (0.95, 0.05) & 0.0089 & 0.0028 & 0.010 & 0.499\ [0.495, 0.502] & 0.501 \\
uninformative   & 3 & (0.50, 0.50) & 0.0161 & 0.0160 & 0.443 & 0.495\ [0.480, 0.509] & 0.500 \\
six-arm         & 6 & (0.60, 0.30) & 0.0106 & 0.0102 & 0.408 & 0.503\ [0.497, 0.508] & 0.501 \\
\bottomrule
\end{tabular}
\end{table}

\paragraph{Statistical Analysis of Trace Invariance.}
Table~\ref{tab:nonid-appendix} summarizes classification and distribution-distance metrics across six environments. To handle the combinatorial state space $K^T$ at $T=500$, total variation (TV) distances are computed over empirical histograms of summary action statistics (per-arm pull ratios and transition counts). Across all regimes:
\begin{enumerate}
    \item Estimated TV distances remain negligible ($\mathrm{TV} \le 0.0161$) and hover close to the finite-sample sampling noise level $\mathrm{TV}_0$.
    \item Out-of-sample accuracy for trained non-linear classifiers ($\mathrm{Acc}_{\text{clf}}$) and Monte Carlo approximations of the Bayes-optimal rule ($\mathrm{Acc}_{\text{Bayes}}$) remain centered at random chance ($0.500 \pm 0.005$).
    \item Unadjusted permutation tests yield two nominal instances below $\alpha=0.05$. However, after Holm--Bonferroni correction for $6$ hypothesis tests, no configuration rejects the null hypothesis of distributional identity at $\alpha=0.05$, consistent with Theorem~\ref{thm:non-identifiability}.
\end{enumerate}

\subsubsection{Individual Belief Trajectories and Sensitivity Analysis}

\begin{figure}[htbp]
\centering
\includegraphics[width=0.7\linewidth]{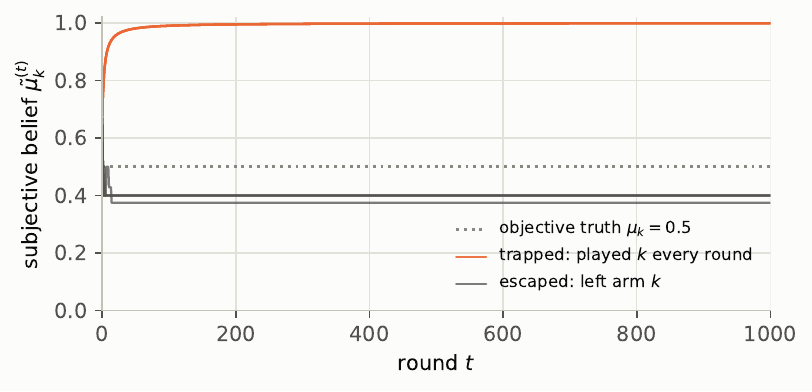}
\caption{\textbf{Individual subjective belief trajectories} $\tilde\mu_k^{(t)}$ for the diagnostic arm ($k=0$, true mean $\mu_k=0.5$, horizontal dotted line) over $T=1{,}000$ rounds. Orange trajectories represent agents locked in the confirmation trap ($1{,}824 / 4{,}000$ runs, of which $1{,}643$ satisfy event $\mathcal{E}_k$ of Theorem~\ref{thm:lock-in}). Grey trajectories depict agents escaping arm $k$, whose subjective beliefs remain constant once play on arm $k$ ceases.}
\label{fig:belief-traj}
\end{figure}

\paragraph{Individual Belief Dynamics.} 
Figure~\ref{fig:belief-traj} illustrates the sample-path behavior underlying Theorem~\ref{thm:lock-in}. Trajectories entering the confirmation trap display subjective expectations that rapidly approach certainty ($\tilde\mu_k^{(t)} \to 1$), despite the true reward mean being $\mu_k=0.5$. This mirrors Step~5 of the proof of Theorem~\ref{thm:lock-in}, where the subjective posterior ratio satisfies $(\alpha_k+t)/(\alpha_k+\beta_k+t) \to 1$ on $\mathcal{E}_k$. In contrast, agents escaping arm $k$ update their subjective expectations downward during initial rounds before abandoning the arm, after which their posterior estimates $\tilde\mu_k^{(t)}$ freeze below the objective mean $\mu_k=0.5$.

\paragraph{Sensitivity to Bias Severity.} 
Figure~\ref{fig:headline} (right) in the main text evaluates performance across $C \in [0.03, 100]$ and $p \in [1.25, 4]$. Within the certified region to the left of the dashed boundary, empirical separation times remain stable: $\approx 27$ rounds under non-coercive execution and $\approx 17$ rounds under coercive blocks, confirming that the horizon $m_\delta$ is robust to distortion strength $(C, p)$ within theoretical bounds. Active bias identification functions empirically across parameter regimes extending up to two orders of magnitude past the analytical boundary. Beyond this empirical threshold, identification degrades. When $C(1-\tilde\mu)^p \ge 1$, the upper bound in Definition~\ref{def:conf_bias} becomes uninformative, leaving only the $\gamma\tilde\mu$ lower bound from Equation~\ref{eq:attenuation}. Under continuous forced failures, the protected prior expectation $\tilde\mu$ diminishes, causing the subjective perception of a severely biased agent to align with that of an unbiased decision-maker and eroding distribution separation.

\begin{figure}[htbp]
\centering
\includegraphics[width=\linewidth]{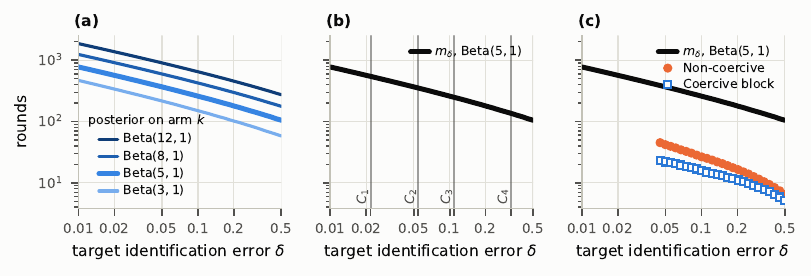}
\caption{\textbf{Analytical sample complexity $m_\delta$ computed via Theorem~\ref{thm:identifiability}.} \textbf{(a)} Required intervention length $m_\delta$ across target error rates $\delta \in [10^{-2}, 0.5]$ for diagnostic arm priors $\mathrm{Beta}(A,1)$ with $A \in \{3,5,8,12\}$. \textbf{(b)} Truncation thresholds $\delta_{\min}(C) = 6\varepsilon_\infty(A,B,C,p)$ induced by distortion parameters $C_1=0.02, C_2=0.05, C_3=0.1, C_4=0.3$ under $p=2$ (vertical lines). \textbf{(c)} Theoretical sample bound $m_\delta$ for $\mathrm{Beta}(5,1)$ compared against actual empirical steps required to achieve target separation under coercive (blue squares) and non-coercive execution (orange circles) ($20{,}000$ seeds, $C=0.005$).}
\label{fig:sample-complexity}
\end{figure}

\paragraph{Sample Complexity and Theoretical Tightness.}
Figure~\ref{fig:sample-complexity} details the behavior of the analytical bound $m_\delta$ established in Theorem~\ref{thm:identifiability}. Panel~(a) shows that stronger initial prior beliefs (higher $A$ in $\mathrm{Beta}(A,1)$) shift the required intervention horizon upward, reflecting the larger number of forced failures needed to overcome higher subjective expectations. Panel~(b) illustrates how increasing the distortion parameter $C$ induces truncation thresholds $\delta_{\min}(C)$ (vertical gray lines $C_1$ through $C_4$), below which arbitrary identification confidence $\delta$ cannot be guaranteed due to persistent asymptotic bias. Finally, Panel~(c) compares the theoretical upper bound $m_\delta$ against empirical separation rounds. Both coercive and non-coercive execution achieve target error rates $\delta$ in substantially fewer rounds than guaranteed by $m_\delta$, confirming that the analytical bound provides a valid and conservative safety envelope for active intervention protocols.

\newpage
\subsection{Grounding in Real Language Models}
\label{sec:exp-llm}

The theoretical guarantees validated assuming an optimal observer estimating marginal action likelihoods $P(a_t=k \mid h_t^a, z)$ via closed-form Beta updating. In practical deployments, interaction monitoring is performed by Large Language Models operating over textual interaction histories, and human decision-makers modeled by Definition~\ref{def:conf_bias} interact directly with these system logs. 

To evaluate whether our theoretical framework transfers to neural architectures, we conduct empirical grounding experiments using \texttt{Qwen3-4B} (executed via vLLM, with reasoning features enabled, a budget of $4{,}096$ tokens, and temperature zero for observer evaluation).

\subsubsection{Active Identification Budget and Slack Sweeps}
\label{sec:app-active-table}

Table~\ref{tab:cost-appendix} provides the complete numerical evaluation of active identification policies across target error thresholds $\delta$ and confirmation slack parameters $C$, extending the summary presented in \S\ref{sec:exp-active}.

\begin{table}[htbp]
\centering
\small
\caption{\textbf{Performance metrics for active identification algorithms across target error tolerances $\delta$ and confirmation slack severity $C$} ($20{,}000$ evaluation runs per configuration). Policies: non-coercive (N-C) vs.\ coercive block (C). $m_\delta$ represents the closed-form analytical bound from Theorem~\ref{thm:identifiability}. Empirical rounds and accuracy record the exact point where population-level Total Variation reaches target $1-\delta$. Forced failures measure the expected number of zero-reward interventions delivered to an unbiased user. Results labeled 'n/r' indicate instances where target separation was not reached within the $600$-round horizon budget.}
\label{tab:cost-appendix}
\begin{tabular}{ccl cccc}
\toprule
\makecell[b]{Confirmation \\ Slack $C$} & \makecell[b]{Target Error $\delta$} & \makecell[b]{Policy} & \makecell[b]{Bound $m_\delta$} & \makecell[b]{Empirical \\ Rounds} & \makecell[b]{Accuracy} & \makecell[b]{Forced Failures \\ (Unbiased)} \\
\midrule
0.01 & 0.5 & N-C & --  & 6.4  & 0.750 & 4.8  \\
0.01 & 0.5 & C   & 105 & 5.2  & 0.750 & 5.2  \\
0.01 & 0.2 & N-C & --  & 16.6 & 0.900 & 8.1  \\
0.01 & 0.2 & C   & 180 & 11.2 & 0.900 & 11.2 \\
0.01 & 0.1 & N-C & --  & 27.0 & 0.950 & 9.6  \\
0.01 & 0.1 & C   & 260 & 16.1 & 0.950 & 16.1 \\
\midrule
3.0  & 0.1 & N-C & --  & 33.9 & 0.950 & 10.2 \\
3.0  & 0.1 & C   & --  & 18.0 & 0.950 & 18.0 \\
\midrule
4.0  & 0.1 & N-C & --  & n/r  & n/r   & n/r   \\
4.0  & 0.1 & C   & --  & 23.7 & 0.950 & 23.7 \\
\bottomrule
\end{tabular}
\end{table}

\paragraph{Detailed Cost and Error Trade-offs.}
As detailed in Table~\ref{tab:cost-appendix}, across all valid confirmation regimes ($C \le 3.0$), the non-coercive identification policy systematically reduces collateral disruption imposed on unbiased users compared to coercive testing blocks. For instance, at a loose tolerance threshold $\delta=0.5$, non-coercive testing requires only $4.8$ forced failures on average, compared to $5.2$ for coercive execution. As the error tolerance sharpens ($\delta=0.1$), this gap widens significantly, saving $6.5$ forced failures per user ($9.6$ vs.\ $16.1$). Furthermore, the table illustrates the theoretical breakdown of non-coercive identification when confirmation slack becomes large ($C=4.0$), where conditional interventions fail to accumulate sufficient TV separation within the execution window ('n/r').

\subsubsection{Does a Real LLM Face the Passive Identification Wall?}

\begin{figure}[htbp]
\centering
\includegraphics[width=0.85\linewidth]{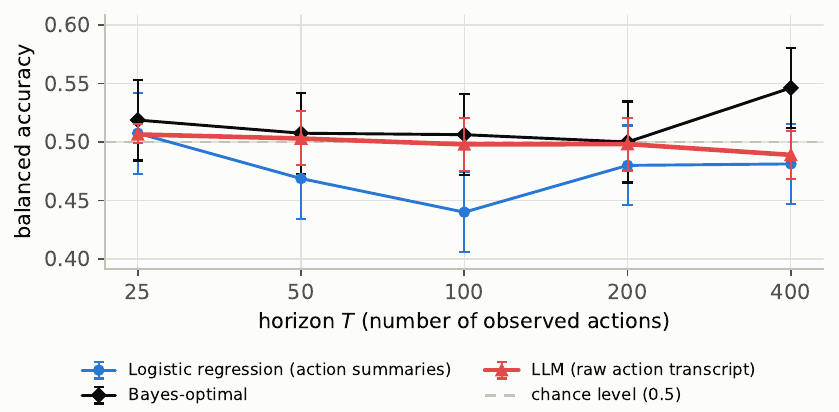}
\caption{\textbf{Passive identification performance on raw action transcripts.} Balanced classification accuracy of \texttt{Qwen3-4B} acting as an action-only observer (red curve) on the baseline environment configuration of Theorem~\ref{thm:non-identifiability}, evaluated alongside logistic regression (blue curve) and the Bayes-optimal decision limit (black curve) across interaction horizons $T \in [25, 400]$ ($400$ trajectories per class, $95\%$ stratified bootstrap CIs). The horizontal dashed line denotes random chance ($0.50$).}
\label{fig:llm-passive}
\end{figure}

\paragraph{Passive Observation Benchmark.} 
We evaluate \texttt{Qwen3-4B} by providing raw action histories $h_T^a$ and prompting the model to classify user type $z \in \{0, 1\}$. Figure~\ref{fig:llm-passive} summarizes balanced accuracy across horizons $T \in [25, 400]$. Across all evaluated horizons, \texttt{Qwen3-4B}'s accuracy remains tightly bounded around random chance ($0.488$ to $0.507$), mirroring the logistic regression baseline and confirming that the theoretical non-identifiability established in Theorem~\ref{thm:non-identifiability} holds strictly for neural observers evaluating unperturbed interaction transcripts.

\subsubsection{Active Identification and Cognitive Bias in LLMs}

\begin{figure}[htbp]
\centering
\includegraphics[width=\linewidth]{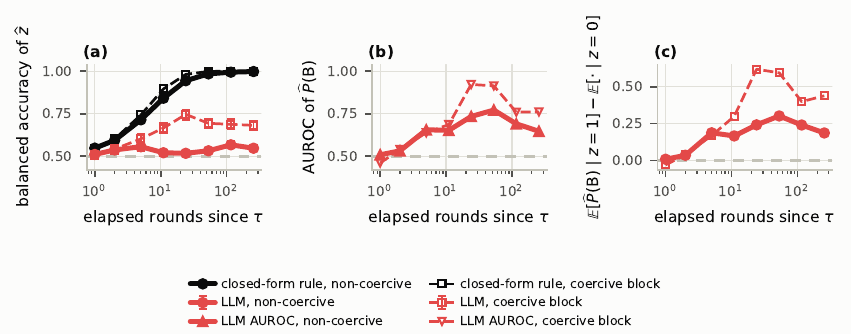}
\caption{\textbf{Evaluation of \texttt{Qwen3-4B} as an active observer receiving action logs and intervention histories} ($9{,}600$ prompts across budgets up to $256$ rounds). \textbf{(a)} Balanced accuracy of binary predictions $\hat z$ vs.\ closed-form reference rules. \textbf{(b)} Area Under the ROC curve (AUROC) derived from model confidence scores. \textbf{(c)} Expected score separation $\mathbb{E}[\widehat P(z{=}1) \mid z{=}1] - \mathbb{E}[\widehat P(z{=}1) \mid z{=}0]$.}
\label{fig:llm-active}
\end{figure}

\paragraph{Active Intervention Monitoring.} 
Figure~\ref{fig:llm-active} evaluates \texttt{Qwen3-4B} in active identification environments where prompts include action choices and intervention logs. As shown in Figure~\ref{fig:llm-active}a, hard binary classifications by the language model achieve accuracy plateaus of $0.55$ under non-coercive monitoring and $0.68$ under coercive blocks at horizon $m=256$. However, Figures~\ref{fig:llm-active}b and~\ref{fig:llm-active}c reveal that raw confidence scores extract substantial discriminative signal: the continuous AUROC peaks at $0.92$ under coercive interventions (at $T=50$) and $0.77$ under non-coercive execution, with score separation reaching $0.60$. This confirms that active interventions supply identifiable signals to neural observers, though full utility requires continuous score calibration rather than rigid output thresholding.

\begin{figure}[htbp]
\centering
\includegraphics[width=\linewidth]{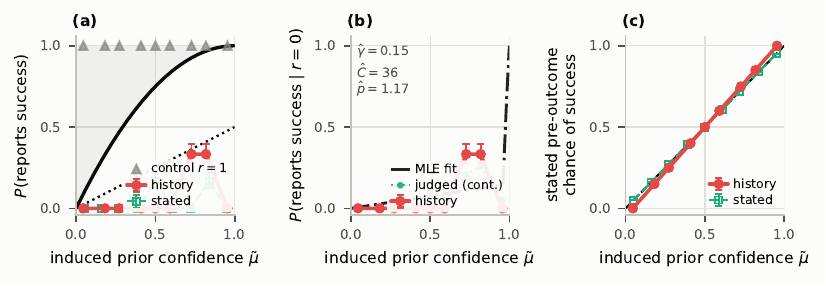}
\caption{\textbf{Elicited cognitive distortion curve for \texttt{Qwen3-4B} acting as a synthetic user} across a $9$-point confidence grid $\tilde\mu \in [0.05, 0.95]$ ($240$ evaluations per point, $95\%$ Wilson CIs). \textbf{(a)} Empirical distortion $\widehat P(\text{reports success} \mid r=0, \tilde\mu)$ under narrative history (red) vs.\ stated verbal induction (green) compared against control $r=1$ (grey triangles). \textbf{(b)} Maximum-likelihood fit of canonical kernel parameters ($\hat\gamma=0.15, \hat C=36, \hat p=1.17$). \textbf{(c)} Manipulation check validating induced prior expectations against reported pre-outcome confidence.}
\label{fig:elicitation}
\end{figure}

\paragraph{Elicitation of Cognitive Distortion.} 
To test whether neural language models exhibit the theoretical confirmation bias formalized in Definition~\ref{def:conf_bias}, we measure the empirical distortion function $\widehat P(\text{reports success} \mid r=0, \tilde\mu)$ using \texttt{Qwen3-4B} as a synthetic user under two prior induction modes: narrative track record (\emph{history}) and direct verbal specification (\emph{stated}).

Figure~\ref{fig:elicitation}c confirms induction fidelity, with pre-outcome reported confidence tracking target priors almost perfectly (Pearson correlation $r > 0.99$). Under true positive outcomes ($r=1$), reporting integrity remains at $100\%$ (grey control triangles in Figure~\ref{fig:elicitation}a). Under objective failure outcomes ($r=0$), empirical success reporting under narrative history induction increases monotonically with prior confidence, peaking at $35\%$ for $\tilde\mu=0.83$ ($z=12.5, p=4.5\times10^{-36}$, Cochran--Armitage trend test). Figure~\ref{fig:elicitation}b presents the maximum-likelihood fit of our canonical family ($\hat\gamma=0.15, \hat C=36, \hat p=1.17$), demonstrating that high prior expectations induce systematic self-confirming reporting distortions in LLM decision agents in close agreement with Definition~\ref{def:conf_bias}.
\section{Population Information and Future Directions}\label{sec:app-population}

The deeper open problem lies in populations. With many users sharing an unknown environment, batch observation may recover the environment while passive data may classify most users, reserving intervention for the ambiguous few.
Indeed, with an unrestricted distribution of environments, a biased population can still remain observationally equivalent to an unbiased population under a transformed environment distribution. 
Even with those constraints, an adaptive policy that optimizes information gained per unit of user burden is a natural bridge from our worst-case construction to deployable elicitation.


\end{document}